\documentclass[a4paper,11pt]{article}
\usepackage{amsmath,a4wide}
\usepackage{amssymb}
\usepackage{amscd}
\usepackage{bm}
\usepackage{epsfig}
\usepackage{enumitem}
\usepackage{graphics}
\usepackage{graphicx}
\usepackage{float}
\usepackage{epsfig}
\usepackage{amsmath}
\usepackage{graphics}
\usepackage{graphicx}
\usepackage{color}
\usepackage{epstopdf}
\usepackage{color}
\usepackage{booktabs}
\usepackage{subfigure}
\usepackage{ifpdf}
\usepackage{amssymb}
\usepackage{cite}
\usepackage{amssymb}
\usepackage{amsthm}
\usepackage{amsmath}
\usepackage[justification=centering]{caption}
\usepackage[]{natbib}
\usepackage[titletoc]{appendix}
\usepackage{verbatim}

\newtheorem{theorem}{Theorem}
\newtheorem{lemma}{Lemma}
\newtheorem{pro}{Proposition}
\newtheorem{cor}{Corollary}
\newtheorem{remark}{Remark}
\newtheorem{assumption}{Assumption}

\begin{document}
	\title{Optimal Expansion of Business Opportunity}
	
	\author{Ling Wang\thanks{Department of Statistics, The Chinese University of Hong Kong, Shatin, N.T., Hong Kong. \newline({\tt lingwang@link.cuhk.edu.hk})}
		\and Kexin Chen\thanks{Department of Applied Mathematics, The Hong Kong Polytechnic University, Hong Kong. (\tt{kexin-neme.chen@polyu.edu.hk})}
		\and Mei Choi Chiu\thanks{Department of Mathematics \& Information Technology, The Education University of Hong Kong, Tai Po, N.T., Hong Kong. \newline({\tt mcchiu@eduhk.hk})}	 \and Hoi Ying Wong\thanks{Corresponding author. Department of Statistics, The Chinese University of Hong Kong, Shatin, N.T., Hong Kong. \newline({\tt hywong@cuhk.edu.hk})}
	}
	\date{\today}
	\maketitle                  
	
	\begin{abstract}
		Any firm whose business strategy has an exposure constraint that limits its potential gain naturally considers expansion, as this can increase its exposure. We model business expansion as an enlargement of the opportunity set for business policies. However, expansion is irreversible and has an opportunity cost attached. We use the expected optimization of utility to formulate this as a novel stochastic control problem combined with an optimal stopping time, and we derive an explicit solution for exponential utility. We apply the framework to an investment and a reinsurance scenario. In the investment problem, the cost and incentives to increase the trading exposure are analyzed, while the optimal timing for an insurer to launch its reinsurance business is investigated in the reinsurance problem. Our model predicts that the additional income gained through business expansion is the key incentive for a decision to expand. Interestingly, companies may have this incentive but are likely to wait for a period of time before expanding, although situations of zero opportunity cost or specific restrictive conditions on the model parameters are exceptions to waiting. The business policy remains on the boundary of the opportunity set before expansion during the waiting period. The length of the waiting period is related to the opportunity cost, return, and risk of the expanded business.   \\
		
		
	\end{abstract}
	
	\newpage
	
\section{Introduction}
Financial companies make decisions about how to control their risk exposure optimally through business activities aimed at maximizing or minimizing various objective functions. Such decisions have recently been the focus of numerous studies (e.g., \cite{Zhou2003}, \cite{He2015}, \cite{NS2020}). In particular, the literature (\cite{OR2008} and \cite{NS2020}) has considered the risk-taking restrictions that commonly occur in the business activities opportunity set due to institutional or regulatory reasons. Commercial banks are typically constrained by domestic regulations aimed at limiting their risk-taking \citep{JoF2012}, and investment firms commonly control the discretion of their traders by placing limits on the risk of their trading portfolios \citep{OR2008}. As demonstrated in  \cite{OR2008}, these risk limits naturally translate into position limits that consider both the risk and the expected return of the position. This introduces an optimal stochastic control model, in which the controls are constrained to a specific domain.

The willingness of financial firms to reduce restrictions is not surprising, even if it comes at a cost. They may gain larger opportunity sets and subsequent profits, which can balance the costs. The internationalization of commercial banks is a typical scenario.  Most regulations that constrain a bank's risk are implemented nationally, but banks can engage in regulatory arbitrage through internalization (\cite{JoF2012}, \cite{Ongena2013}).  \cite{JoF2012} find evidence that cross-county differences in regulations encourage the flow of bank capital from more to less restrictive jurisdictions. Given the interconnected nature of financial markets and institutions, these regulatory differences enable banks to circumvent strict domestic regulations and take greater risks abroad \citep{Barth2008}.  Evidence has also been presented \citep{Galema2016} that banks internationalize to achieve higher returns through taking more risks, which implies that this strategy leads to potential financial benefits that can offset the cost of internationalization. These findings raise the question of how financial firms can identify the optimal time to mitigate restrictions (or for banks to internationalize), which we address in this paper. 

 We propose a novel model for investigating the optimal time to relax the constraint and expand a firm’s business. Consider a firm whose business activities are modeled by a control process $f_t$, $t\ge0$, which initially takes on values in $\mathcal{D}$, with risk and potential profit at any time $t$ proportional to $f_t$. Thus, various business activities simultaneously change the company's risk and its potential profit. The restriction $f \in \mathcal{D}$ reflects the risk control constraint before business expansion. The opportunity set for business policies $\mathcal{D}$ is enlarged after expansion due to the relaxing of restrictions.  When the expansion decision is irreversible and a running cost is subsequently required, the problem is mathematically formulated as an optimal control-stopping problem. 
 
 Our model can be applied to other areas, particularly insurance. Insurance companies commonly apply a proportional reinsurance strategy to divert a proportion of their risk, which has been extensively examined in the literature (e.g., \cite{GL},  \cite{Zeng2016},  \cite{LRZ2014}, \cite{Wei2019},  \cite{YW2020}). 
 The reinsurance control $f_t$ falls strictly inside the range of $[0,1]$ if the insurer purchases the reinsurance.  However, to remove the constrained opportunity set $[0,1]$, many studies simply allow the resulting optimal control to go beyond 1. This, however, is interpreted as the situation that the insurer acquiring a new business, which is somewhat vague.  As the control $f$ represents the reinsurance strategies, the new business is essentially providing reinsurance protection. The insurer takes on more risks to obtain higher return by expanding the business into the reinsurance sector. We apply the business expansion framework to this context rather than the conventional setting and assess the optimal timing required to expand into the reinsurance business. Before the expansion, the reinsurance opportunity set is $[0,1]$; it is enlarged after expansion to the positive real line, and an opportunity cost $\rho$ is charged for operating the new business.

Examining the proposed optimal expanding business problem represents a valuable contribution to the literature. The problem of characterizing the optimal stopping rule that changes the control opportunity set has not yet been explored, although the literature on optimal stopping problems in finance and insurance is extensive. In most stopping problem formulations, either the parameters of the state equations change at the optimal stopping time, such as the optimal consumption-investment problem with a retirement option (e.g., \cite{CS2006, DL2010}, \cite{YK2018},  \cite{CJW2020}),  or the objective functions are dependent on the stopping time, such as the American option pricing problem (e.g., \cite{K1990}) and optimal investment with stopping (e.g., \cite{Jian2014}). The admissible control set before and after stopping is identical in both cases.
Unlike these optimal stopping problems, the constrained
control process is much more involved as it is directly affected by the stopping time in our problem. In control-stopping problems with no constraint, the duality approach is commonly taken to separately solve the control and stopping parts (see e.g., \cite{KW2000}, \cite{DL2010},  \cite{YK2018},  \cite{Guan2017},  \cite{MXZ2019}).  Here, the original problem is transformed to a pure optimal stopping dual problem without control, and the dual-state variable is unaffected by the stopping time. With constraint, the interplay between the stopping time and the constraints on the controls induce two sequential effects: (1) they render the duality approach inapplicable because the control is constrained before stopping; and
(2) by applying the dynamic programming principle, two operators in the variational inequality are generated (see equation \eqref{VI0}).  
In this paper, we provide a new method for explicitly solving the mixed control-stopping problem by means of an auxiliary problem in the context of an exponential utility function.  We present our economic interpretations based on the solution.  Our result reveals that a higher expected return is the main motivation for business expansion.  The  level of restrictiveness of the constraint before the expansion is also important in the decision to expand.  The firm will be less motivated to expand as the original regulation is loosened, which coincides with the intuitive impression.  Our model predicts that the firm is highly likely to wait for a period of time even if it has an incentive to expand. This phenomenon is closely related to the opportunity cost. The value of the opportunity cost, the expected return, and the risk of expanding the business all determine the length of the waiting time.

The remainder of this paper is organized as follows. Section \ref{sec:Problem} introduces and formulates the general optimal business expansion problem. In Section \ref{sec:Problem}, the VI associated with the problem is presented, and the verification theorem is proven.   Section \ref{sec:Company} presents the specific context of the optimal investment expansion problem, which involves a stochastic control, an optimal stopping time, and a bounded business constraint.  The solution to the problem is also derived explicitly under an exponential utility function. Comparative statistics and economic explanations are given in  Section \ref{sec:parameter}.   We present the results of the reinsurance expansion problem  in Section \ref{sec:RePro}.   Section \ref{sec:conclude} concludes the paper.

	\section{Formulation of the optimal business expansion problem}
	\label{sec:Problem}
    
    \subsection{A motivating example}
    Consider a bank that operates in a domestic market. Its business activities are restricted by the domestic regulatory system. For example, the minimal capital requirements for market risk impose a restriction on the bank's trading portfolio in the form of a constraint on the level of the bank’s risk exposure.  We denote $\mathcal{D}_1$ as the set of business policies the bank can choose from under the regulations.  
    
    \cite{JoF2012} note that the bank has an incentive to expand its business to a market with less restrictive regulations, but if it decides to expand its operations so that it can potentially acquire greater profits, its exposure to risk increases. The bank can then select its business strategy from an enlarged set of options $\mathcal{D}_2$, where the relationship
    \begin{equation}\label{D12}
	\mathcal{D}_1 \subset \mathcal{D}_2
	\end{equation}
	is satisfied.   However, the bank incurs some costs by entering a new market. Thus, the bank must identify the best time to expand its business and broaden its acceptable set of strategies. 
    
    \subsection{Problem formulation}
     We formulate an optimal stopping problem with a stopping feature to study the above problem, as follows.  A diffusion process is used to model the surplus of a representative company. 
     Suppose that the surplus process $X_t$ follows
	\[ dX_s^{f} = \left( A_sX_s^{f} + B_s^\top f_s + C_s \right)dt + (\sigma_sf_s)^\top dW_s, ~ s \in [0, T], \]
	with the initial value $X_0 \in \mathbb{R}_+$, where $W$ represents $n$-dimensional standard Brownian motion in a probability space $(\Omega, \mathcal{F}, \mathbb{P})$.
	Here, a continuous compound interest rate in the market can explain the parameters $A_s \in \mathbb{R}$, and $C_s \in \mathbb{R}$ indicates the deterministic income or the constant liabilities the firm faces.  The business strategy is denoted by the control $f_s \in \mathbb{R}^n$, and the company then initially manages the business with control $f$ over the domain $\mathcal{D}_1 \subset \mathbb{R}^{n\times n}$. The parameters $B_s\in \mathbb{R}^n$ and  $\sigma_s \in \mathbb{R}^{n\times n}$ appear as the coefficients of control $f$, reflecting the return and risk ratios of the business policy, respectively. In addition, $T \in \mathbb{R}_+$ is a finite time horizon.
	
	In the above model, the firm is only allowed to choose the control $f$ from the domain $\mathcal{D}_1$ before expansion, thus limiting the return or the potential gain.  The firm has the option to expand its business opportunities or, equivalently, the admissible control set for $f$. This control set is enlarged after the expansion to a new domain $\mathcal{D}_2 \subset \mathbb{R}^{n \times n}$, and thus \eqref{D12} holds and the firm is allowed to take on more risks due to the expectation of more profit by choosing a business policy from a larger set.  If it decides to expand its business opportunities at $\tau$, then it must pay an opportunity cost at the rate of $\rho$ as $\tau$. This option enables the firm to choose the best time $\tau$ to expand, which constitutes an optimal stopping feature. 
	
	We formulate the problem as follows.  For any $t \in [0, T]$, the stopping time $\tau\in {\mathcal F}_t$ takes a value in $[t, T]$. Denote  $\mathcal{T}_{t, T} = \{\tau\in {\mathcal F}_t | \tau \in [t,T] \}$ as the set of  all stopping times. The surplus follows an $\mathcal{F}$-adapted stochastic process when the stopping time is incorporated:
	\begin{equation}\label{Xoriginal}
	dX_s^{\tau, f} = \left( A_sX_s^{\tau, f} + B_s^\top f_s + C_s - \rho\boldsymbol{1}_{\{s\ge\tau\}} \right)dt + (\sigma_sf_s)^\top dW_s,~ X_t = x,
	\end{equation}
	where the constant $\rho \in \mathbb{R}_+$ represents the opportunity cost rate.  
	The control $f_t \in \mathcal{D}_1$ for  $t<\tau$ and $f_t\in \mathcal{D}_2$ for $t \ge \tau$.
	
	The firm aims to maximize the expected utility function so it can select the best business expansion option and optimally manage the business:
	\begin{equation}
	J(t, x; \tau, f) = \mathbb{E}_{t, x}\left[U(T, X_T^{\tau, f})\right]
	\end{equation}
	The optimal stopping time $\tau$ and the control $f$ are determined simultaneously,  where $U(t, x) \in C^\infty([0 , T], \mathbb{R})$ represents the company's utility function. A requirement for the pair of decision values $(\tau, f)$ to be admissible is
	\begin{equation}\label{admissible0}
	\int_{0}^{T}\| f_s\|^4 ds < + \infty  \text{  a.s. in } \Omega, ~ \tau \in \mathcal{T}_{t, T}, ~ \mathbb{E}_{t,x}\left[U^-(T, X_T^{\tau, f}) \right] < + \infty, 
	\end{equation}
	where $U^- = \max\{0, -U\}$. Here, $\| \cdot\|$ is the $L^2$ norm.  Denote $\mathcal{A}_0(t, x)$ as the set of all pairs $(\tau, f)$ satisfying \eqref{admissible0}. 
	For any $t \in [0, T]$,  the set of admissible controls is defined by
	\begin{align}\label{admissible1}
	\mathcal{A}_{\tau}^{t}:=&\left\{(\tau, f) \in \mathcal{A}_0(t, x),  \text{ and } f_s: \Omega \times [t,T] \mapsto \mathbb{R}^n, \mathcal{F}_s\text{-adapted},\right. \cr 
	&\left. \text{ with }  f_s \in \mathcal{D}_1,\text{ for } s\in[t,\tau]; f_s \in \mathcal{D}_2,\text{ for } s\in[\tau,T]\right\}.
	\end{align}
	We can therefore write our problem as
	\begin{equation}\label{gen:V}
	V(t, x) \triangleq \sup_{(\tau, f) \in \mathcal{A}_\tau^t}\mathbb{E}_{t, x}\left[U(T, X_T^{\tau, f})\right] = J(t, x; \tau^*, f^*), 
	\end{equation}
	where $(\tau^*, f^*)$ is the optimal stopping time and the control.  
	
	With the restriction $(\tau, f )\in\mathcal{A}_0(t, x)$, the value function can be written into a two-layer iteration:
	\[V(t, x) = \sup_{\tau \in \mathcal{T}_{t, T}, f\in \{(f_s)_{t\leq s < \tau} \in \mathcal{D}_1 \}}\mathbb{E}_t\left[\sup_{f \in\{(f_s)_{\tau \leq s \leq T}\in \mathcal{D}_2\}  }\mathbb{E}_\tau\left[U(T, X_T^{\tau, f})\right]\right].\] 
	The optimal stopping time then interacts with the control, while the various admissible control sets ${\mathcal D}_1\subset {\mathcal D}_2$ before and after the stopping time present additional challenges, compared with the American option pricing and optimal retirement problems. In terms of the optimal expansion of investment, the bounded constraint ${\mathcal D}_1 = [0, \beta]$, with $\beta > 0$ on the investment strategy before the stopping time, introduces a further complication, as the martingale and duality approach appears much less obvious than in \cite{Guan2017}, \cite{LLS2018}, and \cite{YK2018}.
	
	We make the standard assumption for the utility function.
	\begin{assumption}\label{as:U}
		The utility function $U(t, x) \in C^\infty([0 , T], \mathbb{R})$ takes values in $\mathbb{R}$ and is strictly increasing and concave in terms of $x$. 
	\end{assumption}
	
	\subsection{Variational inequality and verification theorem}
	\label{sec:VI}
	We now establish a variational inequality (VI) to analyze Problem \eqref{gen:V}. We first consider an optimization problem for  $t\ge \tau$ so that $f_t \in \mathcal{D}_2$. Here, it becomes the classic utility maximization problem. If we can solve this classic problem we can then apply the dynamic programming principle and investigate the problem for $t<\tau$.
	
	After its business expansion, the firm focuses on management control and no longer needs to consider the stopping time $\tau$. Thus, for $t\ge \tau$, the company's problem, denoted as $V^{(1)}$, is given by
	\begin{align}\label{V(1)}
	V^{(1)} (t,x) : = \sup_{f \in \mathcal{A}_{t}^t }  \mathbb{E}\left[ U (T, X_T^{t, f})  \right]. 
	\end{align}
	From the dynamic programming principle, $V^{(1)}$ satisfies the Hamilton-Jacobi-Bellman (HJB) equation,
	\begin{equation}\label{HJB0}
	\left\{
	\begin{array}{lr}
	- \partial_tV^{(1)} - \max_{f \in \mathcal{D}_2 } \left\{ \mathcal{L}_1^fV^{(1)}\right\} = 0 ~ \text{in}~ \mathcal{N}_T; \\
	V^{(1)}(T, x) = U(T, x),~ \forall x \in \mathbb{R},
	\end{array}
	\right.
	\end{equation}
	where  $\mathcal{N}_T = [0, T]\times \mathbb{R}$, and 
	$$\mathcal{L}_1^f \triangleq \frac{1}{2}f^\top\sigma^\top\sigma f \partial_{xx} + (Ax + B^\top f + C -\rho)\partial_x.$$
	
	The value of $V^{(1)}$ coincides with that of $V$ at $\tau$. For $t<\tau$, we have
	\begin{align*}
	V(t,x) &= \sup_{(\tau, f)   \in \mathcal{A}_{\tau}^t } \mathbb{E}_t\left[ U (T, X_T^{\tau, f})  \right] = \sup_{(\tau, f)   \in \mathcal{A}_{\tau}^t } \mathbb{E}_t\left[\sup_{f \in \mathcal{A}_\tau^\tau}\mathbb{E}_\tau\left[U (T, X_T^{\tau, f}) \right] \right]\\
	& = \sup_{(\tau, f)   \in \mathcal{A}_{\tau}^t } \mathbb{E} \left[ V^{(1)}(\tau, X_{\tau}^{\tau, f})  \right]. 
	\end{align*}
	Again, from the dynamic programming principle, we derive the following VI for the value function $V$ defined in \eqref{gen:V}. 
	\begin{align}\label{VI0}
	\left\{\begin{array}{ll}
	-\partial_tV - \max_{f\in \mathcal{D}_1} \left\{\mathcal{L}^fV\right\} = 0, & \text{if }  V > V^{(1)} \text{ and } (t, x) \in \mathcal{N}_T;\\
	-\partial_tV - \max_{f\in \mathcal{D}_1} \left\{\mathcal{L}^fV\right\} \geq 0, & \text{if } V = V^{(1)} \text{ and } (t, x) \in \mathcal{N}_T;\\
	V(T, x) = U(T,x), & \forall x \in \mathbb{R},
	\end{array}
	\right.
	\end{align}
	where 
	$$\mathcal{L}^f \triangleq \frac{1}{2}f^\top\sigma^\top\sigma f \partial_{xx} + (Ax + B^\top f + C)\partial_x. $$
	 Due to the nonlinear term $\max_{f\in \mathcal{D}_1} \left\{\mathcal{L}^fV\right\}$, the VI \eqref{VI0} is very different from the classical VI, which only consists of a linear operator.
	If the pre-opportunity set ${\mathcal D}_1$ is unbounded and is, for example, a positive real line, a possible approach to solving this is to apply the duality formulation. However, the duality approach is not suitable for our problem as we are broadly encountering a bounded pre-opportunity set. 
	
	The following verification theorem shows the equivalence between the VI \eqref{VI0} and the Problem \eqref{gen:V}. 
	Recall the space $W_p^{2,1}(\Omega) = \{u(y, t) \in L^p(\Omega)| u_{yy}, u_t \in L^p(\Omega)\}$ and $W_{p, loc}^{2,1}(\Omega) = \{u \in W_p^{2, 1}(\Omega')| \text{for any } \Omega' \subset \Omega\}$ defined in \cite{LSU1968}.

	\begin{theorem}\label{thm1}
		Suppose that $v^1$ is the strong solution to the HJB equation \eqref{HJB0} and $v$ the strong solution to VI \eqref{VI0}. Then, suppose that $v$ and $v^1$ have the following properties: 
		\begin{enumerate}[label = (\roman*)]
			\item $v, v^1 \in W^{2,1}_{p, \text{loc}}(\mathcal{N}_T) \cap C(\mathcal{N}_T) $ with some $p > 2$. \label{condition1}
			\item $\partial_{x}v  > 0$ and $\partial_{xx}v < 0$ a.e. in $\mathcal{N}_T$. \label{condition2} 
			\item $\mathbb{E}_{t,x}\left[\int_{t}^{T}(\partial_x{v})^4ds \right]< \infty$, and $\mathbb{E}_{t,x}\left[\int_{t}^{T}(\partial_x{v^1})^4ds \right]< \infty$ for $t \in [0, T]$. \label{condition3'}
			\item The optimal stopping time $\tau = \inf\{s\geq t: v(s, X_s) = v_1(s, X_s)\}$ has a unique solution $\bar{\tau}$, and there is a function $\bar{f}\in \mathcal{A}_{\bar{\tau}}$ such that 
			\begin{align*}
			-\partial_tv - \mathcal{L}^{\bar{f}}v & = -\partial_tv - \max_{f\in \mathcal{D}_1} \left\{\mathcal{L}^fv\right\} = 0, ~\text{if} ~ v > v^1;\\
			- \partial_tv^1 -  \mathcal{L}_1^{\bar{f}}v^1 &= - \partial_tv^1 - \max_{f \in \mathcal{D}_2} \left\{ \mathcal{L}_1^fv^1\right\} = 0, ~\text{if} ~ v = v^1,
			\end{align*}
			where $\mathcal{L}^{\bar{f}} \triangleq \frac{1}{2}\bar{f}^\top\sigma^\top\sigma \bar{f} \partial_{xx} + (Ax + B^\top \bar{f} + C)\partial_x$ and $\mathcal{L}_1^{\bar{f}} \triangleq \frac{1}{2}\bar{f}^\top\sigma^\top\sigma \bar{f} \partial_{xx} + (Ax + B^\top \bar{f} + C - \rho)\partial_x$.  In addition, the following SDE 
			\begin{align*}
			dX_s^{\tau, \bar{f}} = \left( A_sX_s^{\tau, f} + B_s^\top \bar{f}_s + C_s - \rho\boldsymbol{1}_{\{s\ge\tau\}} \right)dt + (\sigma_s\bar{f}_s)^\top dW_s,~ X_t = x,
			\end{align*}
			enables a unique strong solution. \label{condition3}
		\end{enumerate}
		Then, 
		$$v(t, x) = V(t, x),~ \forall (t, x) \in \mathcal{N}_T.$$
		and $(\tau^*, f^*) = (\bar{\tau}, \bar{f})$ is the optimal stopping time and control. 
	\end{theorem}

   \section{Optimal investment expansion}
   \label{sec:Company}
   In this section, we apply our investment expansion problem framework to an investment firm and derive an explicit solution to the problem under an exponential utility function.
   
   Investment firms usually impose limits on the risk of their trading portfolios, due to the regulatory environment.  The value at risk is typically applied as the risk measure, and the risk limit is then equivalent to a constraint on the trading portfolio, as examined in \cite{OR2008}. To achieve simplicity without sacrificing generalization,  we consider a one-dimensional case. Assume that the investor invests in a risk-free bond with interest rate $r\in \mathbb{R}_+$ and a risky asset following the Black-Scholes model with drift rate $\mu + r\in \mathbb{R}_+$ and volatility $\sigma\in \mathbb{R}_+$. 
   The investment firm's surplus then follows the stochastic process:
   \[d X_{s}^{\tau, f}=\left(r X_{s}^{\tau, f}+ \mu f_s  -\delta\right) d s +\sigma f_s d W_s\]
   with an initial value $x_0 \in \mathbb{R}_+$,  where $\delta \in \mathbb{R}_+$ represents the constant liability the firm has to pay per unit of time or the debt rate, and $W$ is a one-dimensional standard Brownian motion.  Here, $f$ is the trading strategy of the firm, representing the investment amount in the risky asset, which affects both the return and risk of the dynamic system.  A greater return is expected when the firm decides to take a greater risk. The control $f$ can only take its value in a bounded set $[0, \beta]$ at the beginning, with $\beta \in \mathbb{R}_+$, reflecting the risk limitation.  Thus, $\mathcal{D}_1 = [0, \beta]$ in this case. 
   
   If the investment firm wants to relax the constraint through business expansion and take on a higher risk to increase the probability of a large return, then an opportunity cost at the rate of $\rho \in \mathbb{R}_+$ would be charged. After expansion, we assume $\mathcal{D}_2 = [0, \infty)$.  For any fixed $t\in [0,T]$, the set of admissible strategies in \eqref{admissible1} becomes
   \begin{align}\label{admissible}
   \mathcal{A}_{\tau}^{t}:=&\left\{(\tau, f) \in \mathcal{A}_0(t, x),  \text{ and } f_s: \Omega \times [t,T] \mapsto \mathbb{R}, \mathcal{F}_s\text{-adapted},\right. \cr 
   &\left. \text{ with }  f_s \in [0,\beta],\text{ for } s\in[t,\tau]; f_s \in [0,+\infty),\text{ for } s\in[\tau,T]\right\}.
   \end{align}  The surplus process then becomes 
   \begin{align}\label{X}
   d X_{s}^{\tau, f}=\left(r X_{s}^{\tau, f}+ \mu f_s  -\delta - \rho \boldsymbol{1}_{\{s\ge\tau\}} \right) d s +\sigma f_s d W_s,~ X_t = x,
   \end{align}
   The state process in \eqref{X}  clearly resembles that in \eqref{Xoriginal} for $n = 1$.    The optimization problem of the investment firm at time $t$ is given by 
   \begin{align}\label{def:v}
   V(t,x) = \sup_{(\tau, f) \in \mathcal{A}_{\tau}^t } \mathbb{E}_{t, x}\left[ U (T, X_T^{\tau, f})  \right], 
   \end{align}
   where $X$ follows \eqref{X}, and $U$ is the utility function that satisfies Assumption \ref{as:U}. 
   \begin{remark}
   	In \cite{Zhou2003}, a firm’s optimal dividend problem is examined, and its business activities are modeled with a control process. The control is required to take on values in the interval $[\alpha, \beta], 0< \alpha < \beta < \infty$.  We, however, impose the constraint $[0, \beta]$ on the control before the stopping (extending) time.
   \end{remark}
  
   For our investment expansion problem, the HJB equation for $V^{(1)}$ in \eqref{HJB0} and VI for $V$ in \eqref{VI0} are derived as follows.
   \begin{equation}\label{HJB}
   \left\{
   \begin{array}{lr}
   - \partial_tV^{(1)} - \max_{f \ge 0} \left\{ \mathcal{L}_1^fV^{(1)}\right\} = 0 ~ \text{in}~ \mathcal{N}_T; \\
   V^{(1)}(T, x) = U(T, x), ~\forall x \in \mathbb{R},
   \end{array}
   \right.
   \end{equation}
   where  $\mathcal{N}_T = [0, T]\times \mathbb{R}$, and 
   $\mathcal{L}_1^f \triangleq \frac{1}{2}\sigma^2f^2\partial_{xx} + (rx + \mu f - \delta -\rho)\partial_x.$
   \begin{align}\label{VI}
   \left\{\begin{array}{ll}
   -\partial_tV - \max_{f\in [0,\beta]} \left\{\mathcal{L}^fV\right\} = 0, & \text{if }  V > V^{(1)} \text{ and } (t, x) \in \mathcal{N}_T;\\
   -\partial_tV - \max_{f\in [0, \beta]} \left\{\mathcal{L}^fV\right\} \geq 0, & \text{if }  V = V^{(1)} \text{ and } (t, x) \in \mathcal{N}_T; \\
   V(T, x) = U(T,x), & \forall x \in \mathbb{R},\\
   \end{array}
   \right.
   \end{align}
   where 
   $\mathcal{L}^f \triangleq \frac{1}{2}\sigma^2f^2\partial_{xx} + (rx + \mu f - \delta)\partial_x.$
   We denote 
   $$\textbf{CR} := \{(t, x): V(t, x) > V^{(1)}(t, x)\}, ~ \textbf{ER} := \{(t, x): V(t, x) = V^{(1)}(t, x)\}.$$
   as the continuation and expansion regions, respectively. 
 
   \begin{lemma}\label{lem0}
   	Under Assumption \ref{as:U}, $V^{(1)}$, the solution to the HJB equation \eqref{HJB}, satisfies
   	\begin{align*}
   	\partial_xV^{(1)} \geq 0, ~ \partial _{xx}V^{(1)} \leq 0. 
   	\end{align*}
   \end{lemma}
   
   \begin{pro}\label{prop:ge1}
   	Suppose that $V^{(1)}$ is strictly increasing and concave. If \textbf{ER} is not empty, then 
   	$$f^* = -\frac{\mu \partial_xV^{(1)}}{\sigma^2 \partial _{xx}V^{(1)}} > \beta, \text{ for } (t, x) \in \textbf{ER}.$$
   \end{pro}
   \begin{proof}
   	The HJB equation \eqref{HJB}  clearly demonstrates that $f^* = -\frac{\mu \partial_xV^{(1)}}{\sigma^2 \partial _{xx}V^{(1)}} $ when $ (t, x) \in \textbf{ER}$.  Suppose there is a nonempty set $I \subset \textbf{ER}$ and $f^* \in [0, \beta]$ on $I$. From VI \eqref{VI}, we have  
   	$$-\partial_tV^{(1)} -  \left\{  \frac{1}{2}\sigma^2{f^*}^2\partial_{xx} V^{(1)}+ (rx +\mu f^* - \delta)\partial_xV^{(1)} \right\} \geq 0$$
   	on $I$. 
   	However, by \eqref{HJB}, 
   	$$  -\partial_tV^{(1)} - \left\{ \frac{1}{2}\sigma^2{f^*}^2\partial_{xx} V^{(1)}+ (rx +\mu f^* - \delta)\partial_xV^{(1)} \right\}  = -\rho \partial_xV^{(1)}  < 0$$
   	on \textbf{ER}, which leads to a contradiction.  The result follows. 
   \end{proof}
   
   Proposition \ref{prop:ge1} suggests that the optimal control is always larger than $\beta$ after expanding the admissible control set. Thus,  the firm will intuitively set the control to be above $\beta$ so it can take on more risks and expect a higher return after the optimal expansion, as it must pay the opportunity cost and the expansion is irreversible. The optimal decision for the company is therefore to break the original principle after the expansion. 
   
   \subsection{Exponential utility function}
   \label{sec:sol}
   To obtain economic interpretations, we consider the exponential utility that clearly satisfies Assumption \ref{as:U}. We consider
   \begin{equation}\label{Uexponential}
   U(t, x) = - \frac{1}{m}e^{-mx},
   \end{equation}
   where the parameter $m> 0$ is the constant risk aversion parameter. 
   
   Under the exponential utility, the explicit solution to $V^{(1)}$ is the classic solution in the literature, which can be summarized as the following lemma.
   \begin{lemma}\label{lem:V(1)}
   	Under the exponential utility function \eqref{Uexponential}, the explicit solution to $V^{(1)}$ defined in \eqref{V(1)} is given by 
   	\begin{align}\label{V(1)exp}
   	V^{(1)}(t,x) = -\frac{1}{m}\exp\left\{ (-mx -d(t))\exp(r(T-t))  + g(t) \right\},
   	\end{align}
   	where $d(t) = -\frac{(\delta + \rho)m}{r}\left(1 - \exp(-r(T-t)) \right)$ and $g(t) = -\frac{\mu^2}{2\sigma^2}(T- t)$, for $t \in [0, T]$.  In addition, the optimal investment strategy in $\textbf{ER}$ is given by 
   	\begin{equation}
   	f^*_t = \frac{\mu}{\sigma^2 m}\exp(-r(T-t)). 
   	\end{equation}
   \end{lemma}

   Proposition \ref{prop:ge1} indicates that it is not optimal to expand the investment once $-\frac{\mu \partial_xV^{(1)}}{\sigma^2 \partial _{xx}V^{(1)}} \leq \beta$. Thus, with exponential utility, Lemma \ref{lem:V(1)} implies that it is never optimal to expand the investment when $ \frac{\mu}{\sigma^2 m} \exp(-r(T-t)) \leq \beta$.  A necessary condition for possible expansion is that $\beta m<\mu/\sigma^2$.   By combining the results in Proposition \ref{prop:ge1} and Lemma \ref{lem:V(1)}, we obtain the following outcome.

   \begin{pro}\label{pro2}
   An expandable business should earn a sufficiently large risk-adjusted return. Specifically, if $\textbf{ER} \neq \emptyset$, then
   	\begin{equation}\label{parameter1}
   	\frac{\mu}{\sigma^2} > \beta m.
   	\end{equation}
   \end{pro}	
   \begin{proof}
   	Clearly, $f^1_t \triangleq \frac{\mu}{\sigma^2 m} \exp(-r(T-t))$ is an increasing function in $t$.  If $\frac{\mu}{\sigma^2} \leq \beta m$,  $f_T^1 = \frac{\mu}{\sigma^2 m} \leq \beta$, and it is never optimal to expand investment by applying Proposition \ref{prop:ge1}.  Hence, \eqref{parameter1} holds if $\textbf{ER} \neq \emptyset$. 
   \end{proof}
   This finding is economically appealing as it demonstrates that the firm will not consider expansion if the business does not earn enough return. 
   Furthermore, a firm is less likely to expand the admissible control set if it is more risk-averse or the trading limit $\beta$ is less tight. 
   
   \begin{lemma}\label{lem:CR}
   An expandable business does not expand until $t_1$, where
   	\begin{equation}\label{t1}
   	t_1 := \inf\{t\in [0, T]: \frac{\mu}{\sigma^2 m} \exp(-r(T-t)) \geq \beta\} \geq 0
   	\end{equation}
   under condition \eqref{parameter1}. In other words,
    $ \{(t,x): 0 \leq t \leq t_1\} \subseteq \textbf{CR}$. 
   \end{lemma}
   \begin{proof}
   	Under the condition \eqref{parameter1}, $f^1_T = \frac{\mu}{\sigma^2m}> \beta$. Thus, there exists $t \in [0, T]$ such that $f^1_t \geq \beta$ and  $t_1$ in \eqref{t1} is well-defined. From the definition of $t_1$, $f_t^1 \in [0, \beta]$ when $t \in [0, t_1]$. The result follows from Proposition \ref{prop:ge1}. 
   \end{proof}
   \begin{remark}\label{rem:largem}
   	If condition \eqref{parameter1} does not hold, $t_1$ defined in \eqref{t1} does not exist, and it is never optimal to expand the investment before the terminal time $T$.
   \end{remark}
   \subsubsection{An auxiliary problem}
   \label{subsec:Simple}
   Although the explicit solution to $V^{(1)}$ can be derived for an exponential utility, it is highly non-trivial to characterize the exercise boundary from the VI \eqref{VI} because of the bounded constraint on $f$ in $\textbf{CR}$. We address this by considering a simplified investment expansion problem as an auxiliary problem, which we can then link to the original problem. 
   
   If $\textbf{ER} \neq \emptyset$, Proposition \ref{prop:ge1} implies that the optimal investment strategy $f^*_t> \beta$ for $(t, x)\in \textbf{ER}$. In addition,  Lemma \ref{lem:CR} ensures that it is not optimal to expand the investment before time $t_1$, i.e., $t_1\le \tau^*$.  Thus, $f^*_t\in [0, \beta]$ for all $t<t_1$.  What happens to the strategy $f^*_t$ for $t\in[t_1, \tau^*)$? A natural conjecture is that $f^* \equiv \beta$ for all $t \in [t_1, \tau^*]$. We therefore construct an auxiliary problem on $\widetilde{\mathcal{N}}_T = [t_1, T] \times \mathbb{R}$ based on this conjecture, as follows.

   In the auxiliary problem, we assume that the control $f$ is fixed at $\beta$ before the extending time $\tau$, and takes its value from the positive real line after $\tau$. 
   Under condition \eqref{parameter1},  for any fixed $t\in [t_1,T]$, the set of admissible pairs of extending time and investment strategy   for the auxiliary problem is given by
   \begin{align}
   \tilde{\mathcal{A}}_{\tau}^{t}:=&\big\{(\tau, f) \in \tilde{\mathcal{A}}_0(t, x),  \text{ and } f_s: \Omega \times [t,T] \mapsto \mathbb{R}, \mathcal{F}_s\text{-adapted}, \cr 
   &\text{ with } f_s \equiv \beta, \text{ for } s\in[t,\tau];~ f_s \in [0,+\infty),\text{ for } s\in[\tau,T]\big\}, 
   \end{align}
   where $\tilde{\mathcal{A}}_0(t,x)$ is the set of the controls satisfying
   \begin{align*}
   \int_{t_1}^{T}\| f_s\|^4 ds < \infty  \text{  a.s. in } \Omega, ~ \tau \in \mathcal{T}_{t, T}, ~ \mathbb{E}_{t,x}\left[U^-(T, X_T^{\tau, f}) \right] < + \infty, 
   \end{align*}
   with $U^- = \max\{0, -U\}$. 
   Define the new value function for the auxiliary problem:
   \begin{align}\label{Vtilde}
   \widetilde{V}(t,x) = \sup_{(\tau, f) \in \tilde{\mathcal{A}}_{\tau}^t } \mathbb{E}\left[ U (T, X_T^{\tau, f})  \right], 
   \end{align}
   for $(t, x) \in [t_1, T]\times \mathbb{R}$. 
   Using the dynamic programming principle, we derive the VI for $\widetilde{V}$: 
   \begin{align}\label{VItilde}
   \left\{\begin{array}{ll}
   -\partial_t\widetilde{V}  - \widetilde{\mathcal{L}}\widetilde{V} = 0, &\text{if } \widetilde{V} > V^{(1)} \text{ and } (t, x) \in \widetilde{\mathcal{N}}_T;\\
   -\partial_t\widetilde{V} - \widetilde{\mathcal{L}}\widetilde{V} \geq 0, & \text{if } \widetilde{V} = V^{(1)} \text{ and } (t, x) \in \widetilde{\mathcal{N}}_T;\\
   \widetilde{V}(T, x) = U(T, x), & \forall x \in \mathbb{R}, 
   \end{array}
   \right.
   \end{align}
   where 
   $\widetilde{\mathcal{L}} \triangleq \frac{1}{2}\beta^2\sigma^2\partial_{xx} + (rx + \beta\mu - \delta)\partial_x$, and $V^{(1)}$ given in Lemma \ref{lem:V(1)} solves the HJB equation \eqref{HJB}.
   For the auxiliary  problem \eqref{Vtilde}, we denote 
   $$\widetilde{\textbf{CR}} := \{(t, x): \widetilde{V}(t, x) > V^{(1)}(t, x)\}, ~ \widetilde{\textbf{ER}} := \{(t, x): \widetilde{V}(t, x) = V^{(1)}(t, x)\}$$
   as the continuation and expansion regions, respectively. 
   
   To characterize the optimal exercise boundary from \eqref{VItilde}, we set 
   $$P = \widetilde{V} - V^{(1)}.$$
   Observe that 
   \begin{align*}
   0 & = \partial_tV^{(1)} + \max_{f \ge 0} \left\{ \mathcal{L}_1^fV^{(1)}\right\} = \partial_tV^{(1)} + (rx + \mu f^1 - \delta - \rho)\partial_xV^{(1)} + \frac{1}{2} \sigma^2(f^1)^2\partial_{xx}V^{(1)}\\
   & =  \partial_tV^{(1)} + (rx + \beta\mu - \delta )\partial_xV^{(1)} + \frac{1}{2}\beta^2\sigma^2\partial_{xx}V^{(1)} \\
   & + \left[(f^1 - \beta)\mu -\rho\right]\partial_{x}V^{(1)} + \frac{1}{2}\sigma^2[(f^1)^2 - \beta^2]\partial_{xx}V^{(1)}\\
   & = \partial_tV^{(1)} + \widetilde{\mathcal{L}}V^{(1)} + \left[(f^1 - \beta)\mu -\rho\right]\partial_{x}V^{(1)} + \frac{1}{2}\sigma^2[(f^1)^2 - \beta^2]\partial_{xx}V^{(1)},
   \end{align*}
   where $f^1 = \frac{\mu}{\sigma^2 m}\exp(-r(T-t))$ by Lemma \ref{lem:V(1)}. $P$ then satisfies
   \begin{align}\label{VI:P}
   \left\{
   \begin{array}{ll}
   \partial_tP + \widetilde{\mathcal{L}}P =  \left[(f^1 - \beta)\mu -\rho\right]\partial_{x}V^{(1)} + \frac{1}{2}\sigma^2[(f^1)^2 -\beta^2]\partial_{xx}V^{(1)}, & \text{if } P> 0 \text{ and } (t, x) \in \widetilde{\mathcal{N}}_T;\\
   \partial_tP + \widetilde{\mathcal{L}}P \leq \left[(f^1 - \beta)\mu -\rho\right]\partial_{x}V^{(1)} + \frac{1}{2}\sigma^2[(f^1)^2 - \beta^2]\partial_{xx}V^{(1)}, & \text{if } P =0 \text{ and } (t, x) \in \widetilde{\mathcal{N}}_T; \\
   P(T, x) = 0, & \forall x \in \mathbb{R}.
   \end{array}
   \right.
   \end{align}
   By Lemma \ref{lem:V(1)}, 
   \begin{align*}
   &\left[(f^1 - \beta)\mu -\rho\right]\partial_{x}V^{(1)} + \frac{1}{2}\sigma^2[(f^1)^2 - \beta^2]\partial_{xx}V^{(1)}\\
   &= \left[\frac{\mu^2}{2\sigma^2 m} + \frac{1}{2}\beta^2\sigma^2 m\exp(2r(T-t)) - (\beta\mu + \rho)\exp(r(T-t))\right]\\
   &\exp\left\{ \left(-mx + \frac{(\delta + \rho)m}{r}\left(1 - \exp(-r(T-t)) \right)\right)\exp(r(T-t))  -\frac{1}{2}\frac{\mu^2}{\sigma^2}(T- t) \right\}\\
   & = h(t)\exp\left\{ \left(-mx + \frac{(\delta + \rho)m}{r}\left(1 - \exp(-r(T-t)) \right)\right)\exp(r(T-t))  -\frac{1}{2}\frac{\mu^2}{\sigma^2}(T- t) \right\}, 
   \end{align*}
   where 
   \begin{equation}\label{h(t)}
   h(t) \triangleq \frac{\mu^2}{2\sigma^2 m} + \frac{1}{2}\beta^2\sigma^2 m\exp(2r(T-t)) -(\beta\mu + \rho)\exp(r(T-t)). 
   \end{equation}
   
   \begin{pro}\label{pro:cond2}
   	If $\widetilde{\textbf{ER}} \neq \emptyset$, then
   	\begin{equation}\label{parameter2}
   \rho \leq \frac{\mu^2}{2\sigma^2 m} + \frac{1}{2}\beta^2\sigma^2 m - \beta\mu. 
   	\end{equation}
   \end{pro}
   \begin{proof}
   	Let $a(t) \triangleq \exp(r(T-t))$, then $h(t) =  \frac{1}{2}\beta^2\sigma^2 m a(t)^2 -(\beta\mu + \rho)a(t) + \frac{\mu^2}{2\sigma^2 m}$ is a quadratic function with respect to $a(t)$. As $a(t)$ is a decreasing function in terms of $t$, from the definition of $t_1$, $a(t) \leq \frac{\mu}{\sigma^2 \beta m}$ when $t \geq t_1$.  Through a simple analysis of $h(t)$, we know that $h(t)$ is an increasing function in $t$ for $t\in [t_1, T]$.  
   	
   	By the VI \eqref{VI:P},  if $P = 0$, we have
   	$$0 \leq h(t)\exp\left\{ \left(-mx + \frac{(\delta + \rho)m}{r}\left(1 - \exp(-r(T-t)) \right)\right)\exp(r(T-t))  -\frac{\mu^2}{2\sigma^2}(T- t) \right\}.$$
   	Thus, $h(t) \geq 0$ when $(t, x)\in \widetilde{\textbf{ER}}$.  It requires $h(T) = 	\frac{\mu^2}{2\sigma^2 m} + \frac{1}{2}\beta^2\sigma^2 m - (\beta\mu + \rho) \geq 0$, if $\widetilde{\textbf{ER}} \neq \emptyset$.  
   \end{proof}
   Based on the above analysis, the following lemma is immediate. 
   \begin{lemma}\label{lemma:t2}
   	Under the conditions \eqref{parameter1} and \eqref{parameter2}, define
   	\begin{equation}\label{t2}
   	t_2 = \inf\{t\in [t_1, T]: h(t) \geq 0\} \geq t_1. 
   	\end{equation}
   	Then, $\{(t, x): t_1\leq t \leq t_2\} \subseteq \widetilde{\textbf{CR}}$. 
   \end{lemma}

   \begin{theorem}\label{thm:tautilde}
   	Under the conditions \eqref{parameter1} and \eqref{parameter2}, the explicit solution to VI \eqref{VI:P} is given by
   	\begin{align}\label{P}
   	P(t, x) =\left\{
   	\begin{array}{lr}
   	-(1- e^{\int_{t}^{t_2}mh(s)ds})V^{(1)}(t,x), \text{ if } t \in [t_1, t_2),\\
   	0, ~ \text{if } t \in [t_2, T], 
   	\end{array}
   	\right.
   	\end{align}
   	where $h(t)$ is defined by \eqref{h(t)}, and $V^{(1)}(t, x)$ is given in Lemma \ref{lem:V(1)}. In addition, $P(t, x)$, $\partial_t P(t, x)$, and $\partial_x P(t, x)$ are continuous functions for $(t, x) \in [t_1, T]\times \mathbb{R}$. The exercise boundary of Problem \eqref{Vtilde} is $ \tau^* = t_2$. 
   \end{theorem}

   As $\widetilde{V} = V + P$, the explicit form of $\widetilde{V}$ is immediate. 
   \begin{cor}\label{corVtilde}
   	The explicit solution to VI \eqref{VItilde} is given by 
   	\begin{align}
   	\widetilde{V}(t, x) =\left\{
   	\begin{array}{lr}
   	-\frac{1}{m}\exp\left\{ \int_{t}^{t_2}mh(s)ds + (-mx -d(t))\exp(r(T-t))  + g(t) \right\}, \text{ if } t \in [t_1, t_2),\\
   	-\frac{1}{m}\exp\left\{ (-mx -d(t))\exp(r(T-t))  + g(t) \right\},  ~ \text{if } t \in [t_2, T],
   	\end{array}
   	\right.
   	\end{align}
   	where $d(t)$ and $g(t)$ are given in Lemma \ref{lem:V(1)}.
   \end{cor}
   
   \subsubsection{Solution}
   In this part we verify that the solution to the auxiliary problem in \eqref{Vtilde} is indeed the solution to the optimal investment expansion problem in \eqref{def:v} for $t > t_1$. The two problems also share the same optimal expansion region.

   \begin{theorem}\label{thm:tau}
   	Under the exponential utility, the optimal extending time for the original problem \eqref{def:v} is the same as that for the auxiliary problem \eqref{Vtilde}. Thus, 
   	\begin{align*}
   	\textbf{ER} = \widetilde{\textbf{ER}},
   	\end{align*}
   	and $V(t,x) = \widetilde{V}(t, x)$ for $t \in [t_1, T]$. 
   \end{theorem}
   \begin{proof}
      See Appendix \ref{proof:thm3}.
   \end{proof}
   Combining Theorem \ref{thm:tau} and Proposition \ref{pro:cond2}, the following property is immediate. 
   \begin{pro}
       The business expansion is possible only when the opportunity cost is not too large. Specifically, if $\textbf{ER} \neq \emptyset$, then condition \eqref{parameter2} holds, i.e., $\rho \leq  \frac{\mu^2}{2\sigma^2 m} + \frac{1}{2}\beta^2\sigma^2 m - \beta\mu $. 
   \end{pro}
   
    This result reveals that the opportunity cost of expanding the business cannot exceed the value $\frac{\mu^2}{2\sigma^2 m} + \frac{1}{2}\beta^2\sigma^2 m - \beta\mu$ or it will never be optimal for the firm to expand the investment and incur the very high cost rate $\rho$.
   
   \begin{remark}
   	If the condition \eqref{parameter2} does not hold,  $t_2$ defined in \eqref{t2} does not exist, and it is never optimal to expand the investment.
   \end{remark}
   
   \begin{pro}\label{pro:control}
   	Under the conditions \eqref{parameter1} and \eqref{parameter2}, the optimal extending time for Problem \eqref{def:v} is 
   	$$\tau^* = t_2.$$
   	The explicit solution to VI \eqref{VI} is given by 
   	\begin{equation}\label{explicitV}
   	V(t, x)=
   	\left\{
   	\begin{array}{lr}
   	-\frac{1}{m}\exp\left\{ (-mx -d_0(t))\exp(r(T-t))  + g_0(t) \right\},~ \text{if } t \in [0, t_1), \\
   	\widetilde{V}(t, x), ~ \text{if } t \in [t_1, T], 
   	\end{array}
   	\right.
   	\end{equation}
   	where $d_0(t) = \left( \frac{(\delta + \rho)m}{r} - \frac{\rho m}{r}\exp(r(T-t_1)) \right)\exp(-r(T-t)) - \frac{\delta m}{r}$, $g_0(t) = -\frac{\mu^2}{2\sigma^2}(T- t) + \int_{t_1}^{t_2}mh(s)ds$,  and $\widetilde{V}$ is given in Corollary  \ref{corVtilde}.  Furthermore, the optimal control for Problem \eqref{def:v} is given by 
   	\begin{align}\label{explicitf}
   	f^*_t = \left\{ 
   	\begin{array}{lr}
   	\frac{\mu}{\sigma^2 m}\exp(-r(T-t)), \text{ if } t\in [0, t_1), \\
   	\beta, \text{ if } t \in [t_1, t_2), \\
   	\frac{\mu}{\sigma^2 m}\exp(-r(T-t)), \text{ if } t\in [t_2, T]. \\
   	\end{array}
   	\right. 
   	\end{align}
   \end{pro}
   
   \begin{proof} The proof follows by a simple substitution of the results for the VI. The details are shown in the Appendix \ref{append:pro5}. 
   \end{proof}
   
   \begin{cor}\label{cor2}
   	The value function $V(t, x)$ in \eqref{explicitV} and $V^{(1)}(t, x)$ in \eqref{V(1)exp} satisfy the conditions in Theorem \ref{thm1}, and the optimal control in \eqref{explicitf} is admissible. 
   \end{cor}

	\section{Economic Interpretations}
	\label{sec:parameter}
	In this section, we offer economic interpretations of the optimal investment expansion problem and analyze the effect of the parameters.
	
	\subsection{Analysis on optimal investment expansion}
	Based on the verification theorem established in Section \ref{sec:Problem}, the explicit solution to the VI \eqref{VI} presented in Proposition \ref{pro:control} is indeed the solution to the optimal investment expansion problem \eqref{def:v}. This includes both the optimal extending time and the control. An illustrative graph of the exercise boundary is given in Figure \ref{fig1}. 
	\begin{figure}[H]
		\centering
		\includegraphics[width = 9cm]{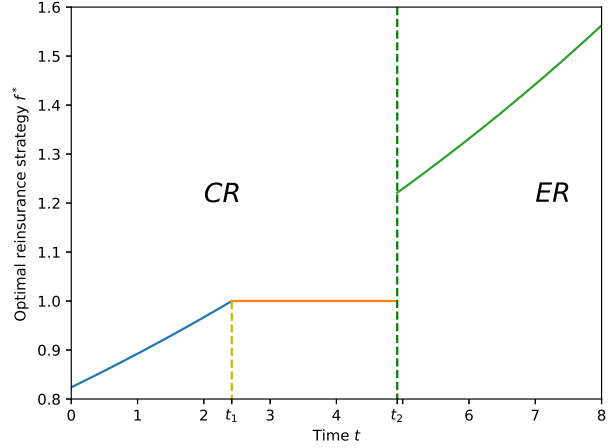}
		\caption{The optimal investment expansion time and investment strategy are computed using $r = 0.08$, $\beta = 1$,  $\mu = 1$, $\rho = 0.04$, $\sigma = 0.8$, $m =1$, and  $T = 8$. }
		\label{fig1}
	\end{figure}

	\begin{pro}
		When the conditions \eqref{parameter1} and \eqref{parameter2} are satisfied,  it is optimal for the firm to expand the admissible control set at a specific time point before the terminal time. 
		In other words, there is always an optimal stopping time $\tau^*\in [0, T]$;  otherwise, it is never optimal for the firm to expand its investment. 
	\end{pro}
	
	
	Three sets of values of $t_1$ and $t_2$, and thus three potential expansion strategies for a company if it is intent on expansion, are possible under the conditions \eqref{parameter1} and \eqref{parameter2}.  We discuss these three cases separately. Note that we assume the conditions \eqref{parameter1} and \eqref{parameter2} are satisfied in the following three cases. 
	
	\subsubsection{Case 1}
	If the parameters satisfy 
	\begin{equation}\label{casecond1}
	\frac{\mu}{\sigma^2}\geq \beta m\exp(rT), \text{ and } \rho \leq \frac{\mu^2}{2\sigma^2 m}\exp(-rT) + \frac{1}{2}\beta^2\sigma^2 m\exp(rT) - \beta \mu, 
	\end{equation}
	then by Lemma \ref{lem:CR} and \ref{lemma:t2}, $t_1 = t_2 = 0$. 
	It is optimal for the firm to expand its investment immediately under \eqref{casecond1}.
	
	Compared with the product of risk averseness $m$ and the original control bound $\beta$,  the risk-adjusted return $\frac{\mu}{\sigma^2}$ is sufficiently large in this case.  Thus,  the firm has an incentive to expand its investment from $t_1 = 0$. 
	The opportunity cost $\rho$ is also sufficiently low, and thus it is optimal for the company to begin expanding and invest more immediately at $\tau^* = t_2 = 0$.

	\begin{pro}\label{pro:case1}
	Under \eqref{casecond1}, if the original upper bound of the control $\beta$ increases, the optimal extending time $\tau^* = 0$ requires a higher risk-adjusted return $\frac{\mu}{\sigma^2}$ and a lower opportunity cost $\rho$. 
	\end{pro}
	By taking the derivative on the values $\beta m\exp(rT)$ and $\frac{\mu^2}{2\sigma^2 m}\exp(-rT) + \frac{1}{2}\beta^2\sigma^2 m\exp(rT) - \beta \mu$, the above proposition immediately follows the equation \eqref{casecond1}. Proposition \ref{pro:case1} reveals that a higher risk-adjusted return and a lower cost are required to initially expand the investment, when the original admissible control set is enlarged. The explanation is that the firm  will not be particularly motivated to expand the admissible control set when the original control set is already large. 
	
	\subsubsection{Case 2}
	If the parameters satisfy 
    \begin{equation}\label{casecond2}
    \frac{\mu}{\sigma^2 }\geq \beta m \exp(rT), \text{ and } \rho >  \frac{\mu^2}{2\sigma^2 m}\exp(-rT) + \frac{1}{2}\beta^2\sigma^2 m\exp(rT) - \beta \mu,
    \end{equation}
	then by Lemma \ref{lem:CR} and \ref{lemma:t2}, $t_1 = 0$, and 
	\begin{equation*}
	 t_2  = T -\frac{1}{r}\ln\left(\frac{(\beta\mu + \rho) - \sqrt{(\beta \mu + \rho)^2 - \beta^2\mu^2}}{\beta^2\sigma^2 m} \right) > 0. 
	\end{equation*}
	Here,  the company never expands the investment until $t_2$ to begin. Meanwhile, the firm uses up the trading limit between $t_1$ and $t_2$, i.e., $f^* \equiv \beta$.
	
	As in Case 1,    $\frac{\mu}{\sigma^2}$ is sufficiently large compared with the product of the risk averseness $m$ and the bound $\beta$.  The company has the incentive to expand the investment from $t_1 = 0$.  
	Although $\rho \leq  \frac{\mu^2}{2\sigma^2 m} + \frac{1}{2}\beta^2\sigma^2 m -\beta\mu$, which is ensured by the condition \eqref{parameter2}, the opportunity cost is not as low as in Case 1.  Thus, it is not optimal for the company to expand its investment at time 0, and the optimal extending time is $\tau^* = t_2 > t_1$. 
	
	If there is no opportunity cost ($\rho = 0$),  the problem is reduced to a standard optimal control problem.  The firm would then expand its investment at $t_1$, i.e., $t_2 = t_1 =0$.  The period in which $f^* \equiv \beta$ in Figure \ref{fig1} would not exist.  However, when $\rho > 0$, the optimal extending time $\tau^* = t_2$ is not necessarily equal to $t_1$, as we have demonstrated.  We therefore refer to $t_2 - t_1$ as the waiting time.   Only when the opportunity cost is sufficiently low, as in Case 1, is the waiting time zero, and the firm expands its investment after receiving the incentive.  
	
	As the opportunity cost is not sufficiently low in this case, the waiting time is
	\begin{equation}\label{wait1}
	t_2 - t_1 =  T -\frac{1}{r}\ln\left(\frac{(\beta \mu + \rho) - \sqrt{(\beta \mu + \rho)^2 - \beta^2\mu^2}}{\beta^2\sigma^2 m} \right).
	\end{equation}
	With the representation of waiting time in \eqref{wait1}, we have 
	\begin{equation*}
	\frac{\partial(t_2 - t_1)}{\partial \rho} = \frac{1}{r \sqrt{(\beta \mu + \rho)^2 - \beta^2 \mu^2}} > 0.
	\end{equation*}
	Thus, the waiting time increases with the opportunity cost. This supports the intuition that a firm becomes cautious when the cost of expanding the business is high but affordable.
	
	\subsubsection{Case 3}
	If  
	\begin{equation}\label{casecond3}
	\frac{\mu}{\sigma^2} < \beta m\exp(rT), 
	\end{equation}
	then by Lemma \ref{lem:CR} and \ref{lemma:t2}, 
	\begin{equation*}
	t_1 = T - \frac{1}{r}\ln(\frac{\mu}{\sigma^2 \beta m}) \geq 0, \text{ and }  t_2  = T -\frac{1}{r}\ln\left(\frac{(\beta\mu + \rho) - \sqrt{(\beta \mu + \rho)^2 - \beta^2\mu^2}}{\beta^2\sigma^2 m} \right) \geq t_1. 
	\end{equation*}
	In this final situation, the firm selects the optimal control in $[0, \beta]$ before $t_1$, uses up the trading limit between $t_1$ and $t_2$ ($f^* \equiv \beta$), and expands at $\tau^* = t_2$.

	Here,  $\frac{\mu}{\sigma^2} < \beta m\exp(rT)$.   Compared with the product of  $m$ and $\beta$, the risk-adjusted return $\frac{\mu}{\sigma^2}$ is not as large as in Cases 1 and 2.  Expanding its investment is not particularly attractive for the firm.  Thus, the company will not have the incentive to expand its investment until $t_1 > 0$.  Even if there were no opportunity cost, the company would wait until time $t_1 = T - \frac{1}{r}\ln(\frac{\mu}{\sigma^2 \beta m}) > 0$ to begin the expansion. When $\rho > 0$, the waiting time is given by 
	\begin{equation}\label{wait2}
	t_2 - t_1 = \frac{1}{r}\ln\left(\frac{(\beta\mu + \rho) + \sqrt{(\beta \mu + \rho)^2 - \beta^2\mu^2}}{\beta \mu} \right).
	\end{equation}
	As in Case 2, we have 
	\begin{equation*}
	\frac{\partial(t_2 - t_1)}{\partial \rho} = \frac{1}{r \sqrt{(\beta \mu + \rho)^2 - \beta^2\mu^2}} > 0.
	\end{equation*}
	We derive the following proposition after these results are combined with those from Case 2.
	\begin{pro}
     Under \eqref{casecond2} or \eqref{casecond3}, the waiting time is an increasing function with respect to the opportunity cost $\rho$. 
	\end{pro}
	An example of how the waiting time varies with the opportunity cost is illustrated in Figure \ref{fig:wt}. 
	\begin{figure}[H]
		\centering
		\includegraphics[width = 9cm]{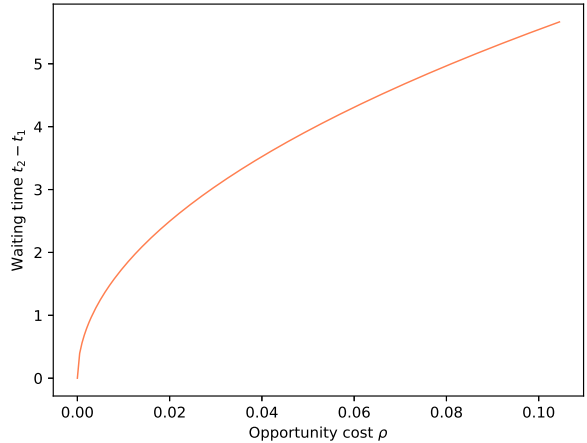}
		\caption{Waiting time with $\rho$ varying in $[0, 0.10]$, $\beta = 1$, $r = 0.08$, $\mu = 0.9$, $\sigma = 0.8$, $m =1$.}
		\label{fig:wt}
	\end{figure}
	From the derivative on $\tau^* = t_2$ with respect to $\beta$, we obtain the following property. 
	\begin{pro}
	Under \eqref{casecond2} or \eqref{casecond3},  the optimal extending time appears later when the original bound $\beta$ grows.
	\end{pro}
	From the above analysis, we find that the risk-adjusted return, the company's risk averseness, the original control bound, and the opportunity cost are all important for the optimal expansion decision.

	\section{Optimal reinsurance business expansion}
	\label{sec:RePro}
	In this section, we apply our problem formulation and results to examine another concrete problem related to reinsurance business expansion. We first review the classic proportional reinsurance model and then link the optimal reinsurance business expansion problem to the investment expansion problem.  
	Although dynamic reinsurance problems have been extensively investigated, it may be criticized that the dynamic reinsurance contracts are unrealistic. However, our result will show that the optimal reinsurance strategy is a deterministic function under the exponential utility function, and thus the insurer can negotiate with the reinsurer and the result is still practically applicable. 
	\subsection{Proportional reinsurance model}
	For a probability space  $(\Omega, \mathcal{F}, \mathbb{P})$ with a filtration $\mathcal{F}_{t}$,  suppose that the insurer's surplus $R_t$ follows the classical Cram\'er-Lundberg model:
	$$
	R_{t} = \hat{x}_0 + c t-\sum_{i=1}^{N_{t}} Z_{i}
	$$
	with the initial surplus  $\hat{x}_0$, in which the arrival process $N_{t}$ is a Poisson process with a constant intensity $\lambda>0$ and the random variables $Z_{i}, i=1,2, \cdots,$ are independent identically distributed claim sizes, independent of $N_{t}$. Let  $G(z)$ denote the claim size distribution with finite first and second moments of $z_1$ and $z_2$, respectively.  The premium rate $c$ is assumed to be calculated via the expected value principle,
	$$
	c=(1+\eta) \lambda z_1, 
	$$
	where $\eta>0$ is the safety loading factor for the insurer.  Assume that  the insurer can manage the insurance risk by purchasing reinsurance. The insurer invests in a risk-free bond with the constant risk-free interest rate $r>0$. We denote the proportional reinsurance strategy by $f_{t}$, $t \in [0, T]$.  The insurer's surplus process then reads
	\begin{align}\label{Xjump}
	d \hat{X}_{t}^{f}=\left\{r\hat{X}_t^f+[ f_t(1+\theta)-(\theta-\eta)] \lambda z_1 \right\}d t-d \sum_{i=1}^{N_{t}} f_{T_{i}} Z_{i}, \quad \hat{X}_0 = \hat{x}_0,
	\end{align}
	where $\theta(\theta \geq \eta)$ represents the safety loading factor for the reinsurer, and  $(1+\theta)(1-f_t) \lambda z_1$ is the premium rate payable to the reinsurance company. 
	
	The conventional positive constraint $f \in [0, \infty)$ is typically adopted in the literature, as in \cite{Zeng2016} and \cite{YW2020}.  $f_t\in[0,1]$ corresponds to a proportional cover. The reinsurance company covers a fraction of the claim, $1-f_{t}$, while the insurer covers the remainder of $f_{t}$.  $f_{t}>1,$ is interpreted as a situation in which the insurer acquires a new business, and thus the insurer herself becomes a reinsurer at some time point without incurring a cost.  The control $f$ measures the risk that the insurer takes on. 
	
	We aim to refine this interpretation. After the reinsurance strategy $f_t$ is allowed to exceed 1 at $t$, the insurer has already expanded her pure insurance business to offer reinsurance protection since $\tau<t$.  Thus, the admissible set is enlarged at $\tau$ and afterward.  In reality, the insurer would bear an opportunity cost for this business expansion.  
	
	\subsection{Reinsurance business expansion}
	\label{sec:rein Problem}

	Recall that $\mathcal{T}_{t, T}$ is the set of all stopping times. These refer to the extending times in our study, and the agent switches from the pure insurance business to reinsurance at $\tau$. For $t<\tau$,  the reinsurance strategy $f_t$ is restricted to the control set $[0, 1]$. For $t\ge\tau$, the insurer expands her pure insurance business to becomes a reinsurance provider so that $f_t\in[0, \infty)$. The insurer then pays the opportunity cost continuously at the rate $\rho$ based on the current surplus level.  Thus, $\mathcal{D}_1 = [0, 1]$ and $\mathcal{D}_2 = [0, \infty)$ in Section \ref{sec:Problem}. 
	When the jump term $d \sum_{i=1}^{N_{t}} f_{T_{i}} Z_{i}$ in the surplus process  \eqref{Xjump} is approximated by a diffusion, as $f_t\left[\lambda z_1 d t - f_t\sqrt{\lambda z_2} d W_{t}\right]$, the approximated surplus process becomes 
     \begin{equation}\label{Xhat}
     d\hat{X}_{s}^{\tau, f}=\left(r \hat{X}_{s}^{\tau, f}+ \mu f_s  -\delta- \rho \boldsymbol{1}_{\{s\ge\tau\}} \right) d s +\sigma f_s d W_s,~ \hat{X}_t = x,
     \end{equation}
	where $\mu = \theta\lambda z_1$, $\delta = (\theta - \eta)\lambda z_1$, and $\sigma = \sqrt{\lambda z_2}$.  We continue to use the exponential utility function in \eqref{Uexponential} for this problem and take a similar approach to defining the insurer's optimization problem as for the optimal investment expansion problem.
	
	Although in reality the problems are different in nature, the optimal reinsurance business expansion problem is mathematically a special case of the problem \eqref{def:v} when $\beta = 1$.  To avoid redundancy, we present the results and economic explanation in the following, rather than discussing the details of the proofs.
     
    \begin{pro}
        If 
        \begin{equation}\label{parameterhat}
          \frac{\mu}{\sigma^2} > m \text{ and }  \rho \leq  \frac{\mu^2}{2\sigma^2 m} + \frac{1}{2}\sigma^2 m - \mu,
        \end{equation}
    then it is  always optimal for the insurer to expand her reinsurance business before the terminal time $T$. Otherwise, it is never optimal for the insurer to expand her business in $[0, T]$. 
    \end{pro} 
    As in the optimal investment expansion problem, condition \eqref{parameterhat} requires the risk-adjusted return to be greater than the risk averseness, and the opportunity cost for running the reinsurance business cannot be too high. Otherwise, it will never be optimal for the agent to expand into the reinsurance business.

    Under the condition \eqref{parameterhat}, define
    	\begin{equation}\label{t1hat}
   	\hat{t}_1 := \inf\{t\in [0, T]: \frac{\mu}{\sigma^2 m} \exp(-r(T-t)) \geq 1 \} \geq 0, 
   	\end{equation}
    and 
    	\begin{equation}\label{t2hat}
    	\hat{t}_2 = \inf\{t\in [\hat{t}_1, T]: \hat{h}(t) \geq 0\} \geq \hat{t}_1,
    	\end{equation}
    	where 
    	\begin{equation}\label{h(t)hat}
    	\hat{h}(t) \triangleq \frac{\mu^2}{2\sigma^2 m} + \frac{1}{2}\sigma^2 m\exp(2r(T-t)) -(\mu + \rho)\exp(r(T-t)). 
    	\end{equation}
    
    Similar to the result in Proposition \ref{pro:control}, the following result is immediate.
    \begin{cor}
    Under the conditions \eqref{parameterhat}, the optimal extending time for the reinsurance expansion problem is $\tau^* = \hat{t}_2$.  The optimal reinsurance strategy is given by
		\begin{align}\label{explicitfhat}
		f^*_t = \left\{ 
		\begin{array}{lr}
		\frac{\mu}{\sigma^2 m}\exp(-r(T-t)), \text{ if } t\in [0, \hat{t}_1), \\
		1, \text{ if } t \in [\hat{t}_1, \hat{t}_2), \\
		\frac{\mu}{\sigma^2 m}\exp(-r(T-t)), \text{ if } t\in [\hat{t}_2, T]. \\
		\end{array}
		\right. 
		\end{align}
    \end{cor}
    
	Under \eqref{parameterhat}, the three results for $\hat{t}_1$ and $\hat{t}_2$ are as follows. 
	\begin{enumerate}[label = (\Roman*)]
		\item\label{case1} 	If 
		$$\frac{\mu}{\sigma^2}\geq m\exp(rT) \text{ and } \rho \leq \frac{\mu^2}{2\sigma^2 m}\exp(-rT) + \frac{1}{2}\sigma^2 m\exp(rT) - \mu,$$
		then 
		$\hat{t}_1 = \hat{t}_2 = 0.$
		
		The insurer should start the reinsurance business immediately.  
		\item\label{case2} If 
		$$\frac{\mu}{\sigma^2 }\geq m \exp(rT) \text{ and } \rho >  \frac{\mu^2}{2\sigma^2 m}\exp(-rT) + \frac{1}{2}\sigma^2 m\exp(rT) - \mu, $$
		then 
		$$\hat{t}_1 = 0 \text{ and } \hat{t}_2  = T -\frac{1}{r}\ln\left(\frac{(\mu + \rho) - \sqrt{(\mu + \rho)^2 - \mu^2}}{\sigma^2 m} \right).$$
		
		Here,  the insurer never purchases reinsurance protection  between $\hat{t}_1$ and $\hat{t}_2$ ($f^* \equiv 1$) and waits until $\hat{t}_2$ to start the reinsurance business. 
		\item\label{case3} If  $\frac{\mu}{\sigma^2}< m\exp(rT),$
		then 
		\begin{equation*}
		\hat{t}_1 = T - \frac{1}{r}\ln(\frac{\mu}{\sigma^2 m}) \text{ and }  \hat{t}_2  = T -\frac{1}{r}\ln\left(\frac{(\mu + \rho) - \sqrt{(\mu + \rho)^2 - \mu^2}}{\sigma^2 m} \right). 
		\end{equation*}
		In this final situation, the insurer purchases reinsurance to manage her insurable risk exposure before $\hat{t}_1$, halts the purchasing between $\hat{t}_1$ and $\hat{t}_2$, and expands at $\tau^* = \hat{t}_2$.		
	\end{enumerate}

	\section{Conclusion}
	\label{sec:conclude}
	Financial firms typically face regulations or constraints on their business activities  and attempt to reduce the power of such restrictions to obtain higher profits.  We take the firms' perspective in this paper and propose a novel optimal business expansion problem to assess the optimal expansion time. We formulate this as an optimal mixed control and stopping problem.  The framework is then applied to both optimal investment and reinsurance business expansion problems. 
	Although many studies focus on optimal portfolio problems with particular limitations,  we contribute to the literature by providing an assessment of the optimal time to relax the constraint.  We also refine the classical proportional reinsurance problem.  
	Although a novel form of mixed control-stopping problem is developed mathematically, we also derive explicit solutions and present economic explanations based on the results.

	Extensive further research can be based on this study.   For example, high-dimensional models can be considered.  We allow the surplus process to be negative, allowing bankruptcy to be considered in future studies.  A one-time setup cost can also be incorporated into the model.    An alternative utility function can certainly be considered, which could be a challenging open problem.   

	\begin{appendix}
		\section{Some Proofs}
		\subsection{Proof of Theorem \ref{thm1}}
			\begin{proof}
		For any  admissible control pair $(\tau, f)\in \mathcal{A}_t^\tau$, by It\^o's formula, 
		\begin{align}\label{ito1}
		v(t, x) = \mathbb{E}_{t, x}\left[v(\tau, X_{\tau}^{\tau, f}) \right] + \mathbb{E}_{t, x}\left[\int_{t}^{\tau}\left(-\partial_tv - \mathcal{L}^fv\right)(s, X_s^{\tau, f})ds\right] + \mathbb{E}_{t, x}\left[-\int_{t}^{\tau}\partial_xv(\sigma_sf_s)^\top dW_s\right].
		\end{align}
		As the condition \ref{condition3'} is satisfied and $f$ is admissible,  the Cauchy-Schwarz inequality determines that 
		\begin{align*}
		\mathbb{E}_{t,x}\left[\int_{t}^{T}(\partial_{x}v)^2(\sigma_sf_s)^\top(\sigma_sf_s) ds\right] \leq C\left\{\mathbb{E}_{t, x}\left[\int_{t}^{T}(\partial_{x}v)^4ds\right]\right\}^{1/2}\left\{\mathbb{E}_{t, x}\left[\int_{t}^{T}\|f_s\|^4ds\right]\right\}^{1/2} < +\infty, 
		\end{align*}
		for a positive constant $C$. Thus, the last term in \eqref{ito1} is a martingale. Then, by the VI \eqref{VI0}, 
		\begin{align*}
		& v(t, x) = \mathbb{E}_{t, x}\left[v(\tau, X_{\tau}^{\tau, f}) \right] + \mathbb{E}_{t, x}\left[\int_{t}^{\tau}\left(-\partial_tv - \mathcal{L}^fv\right)(s, X_s^{\tau, f})ds\right]\\
		&\geq  \mathbb{E}_{t, x}\left[v(\tau, X_{\tau}^{\tau, f}) \right] + \mathbb{E}_{t, x}\left[\int_{t}^{\tau} \left(-\partial_tv - \max_{f\in \mathcal{D}_1} \left\{\mathcal{L}v\right\}\right)(s, X_s^{\tau, f})ds\right]\\
		& \ge \mathbb{E}_{t, x}\left[v^1(\tau, X_{\tau}^{\tau, f}) \right]\\
		& = \mathbb{E}_{t, x}\left[v^1(T, X_T^{\tau, f})\right] + \mathbb{E}_{t, x}\left[\int_{\tau}^{T}\left(-\partial_xv^1 -\mathcal{L}_1^fv^1 \right)(s, X_s^{\tau, f})ds \right]+ \mathbb{E}_{t, x}\left[-\int_{t}^{\tau}\partial_xv^1(\sigma_sf_s)^\top dW_s\right].
		\end{align*}
		Again, by condition \ref{condition3'} and the admissibility of $f$, the last term in the above equation is a martingale, and 
		\begin{align*}
		v(t, x) & \geq \mathbb{E}_{t, x}\left[v^1(T, X_T^{\tau, f})\right] + \mathbb{E}_{t, x}\left[\int_{\tau}^{T}\left( -\partial_tv^1 - \max_{f \in \mathcal{D}_2 } \left\{ \mathcal{L}_1v^1\right\}\right)(s, X_s^{\tau, f})ds \right]\\
		&\geq \mathbb{E}_{t,  x}[U(T, X_T^{\tau, f})]. 
		\end{align*}
		As $(\tau, f)$ is arbitrary, we have $v(t, x) \geq V(t, x)$. 
		
		Similarly, using It\^o's formula, 
		\begin{align*}
		&v(t, x) = \mathbb{E}_{t, x}\left[v(\bar{\tau}, X^{\bar{\tau}, \bar{f}}_{\bar{\tau}}) \right] + \mathbb{E}_{t, x}\left[\int_{t}^{\bar{\tau}} \left(-\partial_tv - \mathcal{L}^{\bar{f}}v\right)(s, X^{\bar{\tau}, \bar{f}}_s)ds\right] + \mathbb{E}_{t, x}\left[-\int_{t}^{\bar{\tau}}\partial_xv(\sigma_s\bar{f}_s)^\top dW_s\right]\\
		& = \mathbb{E}_{t, x}\left[v^1(\bar{\tau}, X^{\bar{\tau}, \bar{f}}_{\bar{\tau}}) \right] \\
		&= \mathbb{E}_{t, x}\left[v^1(T, X^{\bar{\tau}, \bar{f}}_T)\right] + \mathbb{E}_{t, x}\left[\int_{\bar{\tau}}^{T}\left( -\partial_tv^1 - \mathcal{L}_1^{\bar{f}}v^1\right)(s, X^{\bar{\tau}, \bar{f}}_s)ds \right] + \mathbb{E}_{t, x}\left[-\int_{\bar{\tau}}^{T}\partial_xv^1(\sigma_s\bar{f}_s)^\top dW_s\right]\\
		& = \mathbb{E}_{t, x}\left[U(T, X_T^{\bar{\tau}, \bar{f}})\right] \leq V(t, x).
		\end{align*}
	\end{proof}
		
     \subsection{Proof of Lemma \ref{lem0}}
     \begin{proof}
		For any admissible $f \in \mathcal{A}_t^t$ and fixed $x_2> x_1$, we have
		\begin{equation}\label{linear}
		X_T^{x_i, f} = x_ie^{r(T-t)} + \int_{t}^{T}e^{r(T-s)}(\mu f_s -\delta - \rho)ds  + \int_{t}^{T}e^{r(T-s)}f_s\sigma dW_s, ~ i = 1,2. 
		\end{equation}
		The property of the utility function under Assumption \ref{as:U} clearly demonstrates that $V(x, t)$ is a nondecreasing function with respect to $x$. 
		
		For the concavity of $V^{(1)}$, note that $\mathcal{D}_2 = [0, \infty)$ is a convex set. If $f_1, f_2 \in \mathcal{D}_2$ are the admissible controls for $x_1 and x_2$ respectively, then $\alpha f_1 + (1 - \alpha)f_2 \in \mathcal{D}_2$ is admissible for $\alpha x_1 + (1 - \alpha)x_2$ with $\alpha \in [0, 1]$. From the  concavity of $U$ and the linearity of the solution \eqref{linear}, we have 
		\begin{align*}
		&\alpha\mathbb{E}_{t}[U(T, X_T^{x_1, f_1})] + (1 - \alpha) \mathbb{E}_{t}[U(T, X_T^{x_2, f_2})]
		\leq \mathbb{E}_t\left[U(T, \alpha X_T^{x_1, f_1} + (1 - \alpha)X_T^{x_2, f_2}) \right]\\
		&= \mathbb{E}_t\left[U\left(T, X_T^{\alpha x_1 + (1 - \alpha)x_2,~ \alpha f_1 + (1 - \alpha)f_2}\right)\right] \leq V^{(1)}(\alpha x_1 + (1 - \alpha)x_2).
		\end{align*}
		By taking the supremum of $f_1, f_2$ in $\mathcal{D}_2$ on both sides of the above equation, the result follows. 
     \end{proof}
		
		\subsection{Proof of Lemma \ref{lem:V(1)}}
		\begin{proof}
			From \eqref{HJB},
			$$  \partial_tV^{(1)} + \max_{f \ge 0}\left\{ \frac{1}{2}\sigma^2{f}^2\partial_{xx} V^{(1)}+ (rx +\mu f- \delta - \rho)\partial_xV^{(1)} \right\}  =  0, ~ V^{(1)}(t, x) =- \frac{1}{m}e^{-mx}. $$
			Suppose that $\partial_xV^{(1)} > 0$ and $\partial _{xx}V^{(1)} < 0$.  By substituting $f = f^1 \triangleq -\frac{\mu \partial_xV^{(1)}}{\sigma^2 \partial _{xx}V^{(1)}} > 0$ into the above equation, we have
			\begin{equation}\label{PDE1}
			\partial_tV^{(1)}  +  (rx -\delta -\rho)\partial_xV^{(1)} - \frac{1}{2}\frac{\mu^2 \left(\partial_xV^{(1)}\right)^2}{\sigma^2\partial _{xx}V^{(1)}} = 0.
			\end{equation}
			
			Consider the ansatz:
			\begin{align*}
			V^{(1)}(t,x) = -\frac{1}{m}\exp\left\{ (-mx -d(t))\exp(r(T-t))  + g(t) \right\}, 
			\end{align*}
			where $d(t)$ and $g(t)$ are deterministic functions with the terminal conditions $d(T) = 0$ and $g(T ) = 0$, respectively. 
			Then,
			\begin{align*}
			\partial_xV^{(1)} &= -m\exp(r(T-t))V^{(1)}; ~ \partial_{xx}V^{(1)} = m^2\exp(2r(T-t))V^{(1)}; \\
			\partial_tV^{(1)} &= \left\{-r(-mx -d(t))\exp(r(T-t)) -d'(t)\exp(r(T-t)) + g'(t) \right\}V^{(1)}. 
			\end{align*}
			By substituting $\partial_xV^{(1)}, \partial_{xx}V^{(1)}, \partial_tV^{(1)}$ into \eqref{PDE1}, we obtain 
			$$ -\frac{\mu^2}{2\sigma^2} + (\delta + \rho)m\exp(r(T-t)) + rd(t)\exp(r(T-t)) - d'(t)\exp(r(T-t)) + g'(t) = 0. $$
			Hence, 
			\begin{align*}
			\left\{
			\begin{array}{lr}
			-\frac{\mu^2}{2\sigma^2} + g'(t) = 0,\\
			g(T) = 0. 
			\end{array} 
			\right.
			~ \left\{
			\begin{array}{lr}
			(\delta + \rho)m + rd(t) - d'(t) = 0,\\
			d(T) = 0. 
			\end{array} 
			\right.
			\end{align*}
			The explicit solution for $V^{(1)}$ follows. In addition, we have 
			$$f^1_t =  -\frac{\mu \partial_xV^{(1)}}{\sigma^2 \partial _{xx}V^{(1)}} = \frac{\mu}{\sigma^2 m} \exp(- r(T-t)).$$
			Thus, after exercise time $\tau$,  $f^*_t = \frac{\mu}{\sigma^2 m} \exp(-r(T-t)).$
		\end{proof}
		
		\subsection{Proof of Theorem \ref{thm:tautilde}}
		
		\begin{proof}
			Clearly, $P(t, x) \geq 0$ for $t \in [t_1, T]$. We verify that $P$ given in \eqref{P} satisfies \eqref{VI:P} in the following. 
			For $t \in [t_2, T]$, \eqref{VI:P} clearly holds under the condition \eqref{parameter2}.  For $t \in [t_1, t_2)$,  we have 
			\begin{align*}
			&\partial_xP(t, x) = -m\exp(r(T-t))P(t, x), ~ \partial_{xx}P(t, x)  = m^2\exp(2r(T-t))P(t, x)\\
			&\partial_tP(t, x) = -mh(t)e^{\int_{t}^{t_2}mh(s)ds}V^{(1)}(t,x) -(1- e^{\int_{t}^{t_2}mh(s)ds})\partial_tV^{(1)} \\
			&= -mh(t)e^{\int_{t}^{t_2}mh(s)ds}V^{(1)}(t,x) -(1- e^{\int_{t}^{t_2}mh(s)ds})\left[(rmx -(\delta + \rho) m)\exp(r(T-t)) + \frac{\mu^2}{2\sigma^2} \right]V^{(1)}. 
			\end{align*}
			Then,
			\begin{align*}
			&\partial_tP + \widetilde{\mathcal{L}}P  \\
			&= \bigg\{-mh(t)e^{\int_{t}^{t_2}mh(s)ds}V^{(1)}(t,x) -(1- e^{\int_{t}^{t_2}mh(s)ds})\left[(rmx -(\delta + \rho) m)\exp(r(T-t)) + \frac{\mu^2}{2\sigma^2} \right]\\
			& -\sigma^2m^2\exp(2r(T-t))(1- e^{\int_{t}^{t_2}mh(s)ds}) + (rx + \mu -\delta)m\exp(r(T-t))(1- e^{\int_{t}^{t_2}mh(s)ds})\bigg\}V^{(1)}\\
			& = \left[-mh(t)e^{\int_{t}^{t_2}mh(s)ds} -mh(t)(1- e^{\int_{t}^{t_2}mh(s)ds}) \right]V^{(1)} = -mh(t)V^{(1)}. 
			\end{align*}
			Thus, \eqref{VI:P} holds. 
			
			The continuity of $P(t, x)$ is obvious, and $\partial_x P(t, x)$ is then also continuous. For $\partial_t P(t, x)$, we have 
			\begin{align}\label{P'}
			\partial_t P(t, x) =\left\{
			\begin{array}{lr}
			-mh(t)V^{(1)}(t,x) -(1- e^{\int_{t}^{t_2}mh(s)ds})\partial_tV^{(1)}(t,x), \text{ if } t \in [t_1, t_2).\\
			0, ~ \text{if } t \in [t_2, T], 
			\end{array}
			\right.
			\end{align}
			We then consider the following two cases.
			\begin{enumerate}
				\item[(i)] $t_1 = t_2$. Then, $h(t_1) \geq 0$ and $[t_1, t_2)$ is an empty set. So, $P \equiv 0$ for $t \in [t_2, T]$. $\partial_t P(t, x) \equiv 0$ is a continuous function. The exercise boundary is $\tau^* = t_2$. 
				\item[(ii)]$t_1< t_2$. Here, $h(t_1) < 0$ and $h(t_2) = 0$. Thus, $\partial_t P(t, x)$ in \eqref{P'} is still a continuous function. 
			\end{enumerate}
			Therefore, the exercise boundary for Problem \eqref{Vtilde} is $t = t_2$. 
		\end{proof}
		
	\subsection{Proof of Theorem \ref{thm:tau}}
	\label{proof:thm3}
	\begin{proof}
   	For Problem \eqref{def:v}, we represent  the optimal investment strategy in $\textbf{CR}$ by $f^*\in [0, \beta]$ for $t> t_1$.  By \eqref{VI}, $V$ satisfies the VI:
   	\begin{align*}
   	\left\{\begin{array}{lr}
   	-\partial_tV - \left\{\frac{1}{2}\sigma^2{f^*}^2\partial_{xx}V + (rx + \mu f^* - \delta)\partial_xV\right\} = 0, ~\text{if} ~ V > V^{(1)};\\
   	-\partial_tV -  \left\{\frac{1}{2}\sigma^2{f^*}^2\partial_{xx}V + (rx + \mu f^* - \delta)\partial_xV\right\} \geq 0,~ \text{if} ~ V = V^{(1)};\\
   	V(T, x) = U(T,x).
   	\end{array}
   	\right.
   	\end{align*}
   	Similar to the method used in Section \ref{subsec:Simple}, by \eqref{VItilde}, 
   	\begin{align*}
   	& \partial_t\widetilde{V}  + \widetilde{\mathcal{L}}\widetilde{V} = \partial_t\widetilde{V} + \frac{1}{2}\beta^2\sigma^2\partial_{xx}\widetilde{V} + (rx + \beta\mu - \delta)\partial_x\widetilde{V}\\
   	& = \partial_t\widetilde{V} + \frac{1}{2}\sigma^2{f^*}^2\partial_{xx}\widetilde{V} + (rx + \mu f^* - \delta)\partial_x\widetilde{V} + \frac{1}{2}\sigma^2(\beta^2 -{f^*}^2)\partial_{xx}\widetilde{V} + \mu(\beta -f^*)\partial_x\widetilde{V}\\
   	& =  \partial_t\widetilde{V} + \frac{1}{2}\sigma^2{f^*}^2\partial_{xx}\widetilde{V} + (rx + \mu f^* - \delta)\partial_x\widetilde{V} + (\beta - f^*)\left\{\frac{1}{2}\sigma^2(\beta + f^*)\partial_{xx}\widetilde{V} +\mu\partial_x\widetilde{V} \right\}.
   	\end{align*}
   	Let 
   	$$D(t, x) \triangleq (\beta - f^*)\left\{\frac{1}{2}\sigma^2(\beta + f^*)\partial_{xx}\widetilde{V} +\mu\partial_x\widetilde{V}\right\} = (\beta - f^*)\sigma^2\partial_{xx}\widetilde{V}\left\{
   	\frac{1}{2}(\beta + f^*) + \frac{\mu\partial_x\widetilde{V}}{\sigma^2\partial_{xx}\widetilde{V}} \right\}. $$
   	From the explicit form of $\widetilde{V}$ given in Corollary \ref{corVtilde}, 
   	$$ \frac{\mu\partial_x\widetilde{V}}{\sigma^2\partial_{xx}\widetilde{V}} = - \frac{\mu}{\sigma^2 m} \exp(- r(T-t)) \leq  -\beta, \text{ when } t \geq t_1. $$
   	Thus, $ \frac{1}{2}(\beta + f^*) + \frac{\mu\partial_x\widetilde{V}}{\sigma^2\partial_{xx}\widetilde{V}} \leq 0$.  As $\beta - f^* \geq 0$ and $\partial_{xx}\widetilde{V}> 0$, it follows that $D(t, x) \leq 0$. 
   	
   	By the comparison principle for VIs (see, e.g., \cite{F1982}, \cite{YLY}),  we have $\widetilde{V}(t, x) \geq V(t, x)$ for $t \in [t_1, T]$, while it is clear that $\widetilde{V}(t, x) \leq V(t, x)$.  Therefore, $\widetilde{V}(t, x) = V(t, x)$ for $ t \in [t_1, T]$. By Lemma \ref{lem:CR}, the result follows. 
   \end{proof}
	\subsection{Proof of Proposition \ref{pro:control}}
	\label{append:pro5}
		\begin{proof}
			By Lemma \ref{lem:CR} and \eqref{VI},  V satisfies 
			\begin{equation}\label{pdeV}
			\partial_tV +  \max_{f\in [0, \beta]} \left\{ \frac{1}{2}\sigma^2f^2\partial_{xx}V + (rx + \mu f - \delta)\partial_xV\right\} = 0,
			\end{equation}
			when $t \in [0, t_1)$.  For $t \in [0, t_1)$,  assume that 
			$$ V(t, x) =  -\frac{1}{m}\exp\left\{ (-mx -d_0(t))\exp(r(T-t))  + g_0(t) \right\}, $$
			for some deterministic functions $d_0$ and $g_0$.  Then, 
			\begin{align*}
			\partial_xV&= -m\exp(r(T-t))V; ~ \partial_{xx}V = m^2\exp(2r(T-t))V; \\
			\partial_tV &= \left\{-r(-mx -d_0(t))\exp(r(T-t)) -d_0'(t)\exp(r(T-t)) + g_0'(t) \right\}V. 
			\end{align*}
			Thus, by the definition of $t_1$, we have 
			$$ 0 < f^*_t = - \frac{\mu \partial_{x}V}{\sigma^2\partial_{xx}V}  =  \frac{\mu}{\sigma^2 m}\exp(-r(T-t)) < \beta,$$
			when $0<t < t_1$.   Substituting $f = f^*_t$ into \eqref{pdeV},  we have 
			\begin{align*}
			\partial_tV +  \frac{\mu^2}{2\sigma^2m^2}\exp(-2r(T-t)) \partial_{xx}V + \left(rx +\frac{\mu^2}{\sigma^2 m}\exp(-r(T-t))  - \delta\right)\partial_xV = 0.
			\end{align*}
			It follows that 
			\begin{align*}
			d_0'(t)  - rd_0(t) - \delta m = 0,\\
			g_0'(t) - \frac{\mu^2}{2\sigma^2} = 0. 
			\end{align*}
			From the continuity of $V$ at $t_1$, we have the terminal condition 
			\begin{align*}
			d_0(t_1) & = d(t_1)= -\frac{(\delta + \rho)m}{r}\left(1 - \exp(-r(T-t_1)) \right), \\
			g_0(t_1) &= g(t_1) = -\frac{\mu^2}{2\sigma^2}(T- t_1) + \int_{t_1}^{t_2}mh(s)ds. 
			\end{align*}
			The explicit form for $V$ follows by solving the two ODEs. From Theorem \ref{thm:tau}, the optimal reinsurance strategy $f^*$ follows. 
		\end{proof}
	\subsection{Proof of Corollary \ref{cor2}}
	\begin{proof}
   	We only need to prove that the condition \ref{condition3'} in Theorem \ref{thm1} is satisfied. 
   	By \eqref{X}, we have
   	\begin{equation*}
   	X_s^{\tau^*, f^*} = xe^{r(s-t)} + \int_{t}^{s}e^{r(s-\nu)}(\mu f^*_\nu -\delta - \rho\boldsymbol{1}_{\{\nu\ge\tau^*\}})d\nu  + \int_{t}^{s}e^{r(s-\nu)}f_\nu\sigma dW_\nu, ~ s \in [t, T]. 
   	\end{equation*}
   	From \eqref{explicitf}, $f^*$ is a bounded deterministic function. From the expression of $V$ in \eqref{explicitV},
   	\begin{align*}
   	& \mathbb{E}_{t, x}\left[\int_{t}^{T} \partial_xV(s, X_s^{\tau^*, f^*})^4ds\right] \leq 
   	C_1\mathbb{E}_{t, x}\left[\int_{t}^{T}e^{-m\exp(r(T-s))X_s^{\tau^*, f^*}}ds\right]\\
   	& = C_1\int_{t}^{T}\mathbb{E}_{t, x}\left[e^{-m\exp(r(T-s))X_s^{\tau^*, f^*}}\right]ds
   	\leq C_1\int_{t}^{T}\mathbb{E}_{t, x}\left[e^{C_2 + \int_t^s -m\exp(r(T -\nu)) f^*_\nu d W_\nu} \right]ds\\
   	& = C_1 \int_{t}^{T}e^{C_2 + \int_{t}^{s}m^2\exp(2r(T-\nu))f^*_\nu d\nu}ds < \infty. 
   	\end{align*}
   	for some positive constants $C_1$ and $C_2$. Similarly, the result for $V^{(1)}$ can be derived out. 
   \end{proof}
		
	\end{appendix}

\end{document}